\documentclass[final]{siamltex}
\usepackage{amsmath}
\usepackage{epsfig,amssymb,latexsym}
\usepackage{amsfonts,psfrag,bbm,color}
\usepackage{cancel}
\usepackage{siunitx}
\usepackage{comment}
\usepackage{mathrsfs}
\usepackage{graphicx}
\usepackage{adjustbox}
\usepackage{float}
\usepackage{textcomp}
\usepackage{multirow}
\usepackage{enumerate}
\usepackage{cancel}
\usepackage{comment}
\usepackage{algpseudocode}
\usepackage{caption}  
\usepackage{subcaption}
\usepackage{url}
\usepackage{rotating}
\usepackage{bigints}
\usepackage{xcolor}
\usepackage{ulem}
\usepackage{mathrsfs}
\usepackage{hyperref}
\usepackage{placeins}

\usepackage{wrapfig}

\newcommand{\jaman}[1]{{\color{blue}#1}}

\usepackage[ruled,vlined]{algorithm2e}
\usepackage{geometry}
\vbadness=\maxdimen

\newcommand{\footremember}[2]{%
    \footnote{#2}
    \newcounter{#1}
    \setcounter{#1}{\value{footnote}}%
}
\newcommand{\footrecall}[1]{%
    \footnotemark[\value{#1}]%
} 
\usepackage{comment}
\usepackage{tikz}
\usetikzlibrary{arrows.meta,positioning,fit,calc,backgrounds,shapes.misc,shapes.geometric}
\usetikzlibrary{shapes.geometric, arrows}
\tikzstyle{startstop} = [rectangle, rounded corners, minimum width=1cm, minimum height=1cm,text centered, draw=black]
\tikzstyle{io} = [trapezium, trapezium left angle=70, trapezium right angle=110, minimum width=1cm, minimum height=1cm, text centered, draw=black, fill=blue!30]
\tikzstyle{method} = [rectangle, rounded corners, minimum width=1cm, minimum height =1cm, text centered, draw=black]
\tikzstyle{process} = [rectangle, minimum width=1cm, minimum height=1cm, text centered, draw=black]
\tikzstyle{decision} = [diamond, minimum width=0.5cm, minimum height=0.5cm, text centered, draw=black, fill=green!30]
\tikzstyle{arrow} = [thick,->,>=stealth]

\usepackage{amscd}

\graphicspath{{./figs/}}

\def\PP{{{\rm l}\kern - .15em {\rm P} }}
\def\PN2{{\PP_{N}-\PP_{N-2}}}

\newcommand{\bu}{\boldsymbol{u}}

\newcommand{\bv}{\boldsymbol{v}}

\newcommand{\bn}{\boldsymbol{n}}

\newcommand{\bV}{\boldsymbol{V}}

\newcommand{\bx}{\boldsymbol{x}}

\newcommand{\deleted}[1]{{}}

\begin{document}
\title{SELF-ADAPTIVE PHYSICS-INFORMED NEURAL NETWORK
FOR FORWARD AND INVERSE PROBLEMS IN HETEROGENEOUS POROUS FLOW}

\author{
Md. Abdul Aziz\footremember{uabm}{D\MakeLowercase{epartment of} M\MakeLowercase{athematics}, U\MakeLowercase{niversity of} A\MakeLowercase{labama at} B\MakeLowercase{irmingham}, AL 35294, USA.}\footremember{support}{P\MakeLowercase{artially} S\MakeLowercase{upported by the} N\MakeLowercase{ational} S\MakeLowercase{cience} F\MakeLowercase{oundation grant} DMS-2425308.}\footnote{C\MakeLowercase{orrespondence: azizm@uab.edu}}%
\and Thilo Strauss\hspace{1mm}\footnote{D\MakeLowercase{epartment of} AI \MakeLowercase{and} A\MakeLowercase{dvanced} C\MakeLowercase{omputing} X\MakeLowercase{i\'an} J\MakeLowercase{iaotong-}L\MakeLowercase{iverool} U\MakeLowercase{niversity}, S\MakeLowercase{uzhou}, 215000, J\MakeLowercase{iangsu}, C\MakeLowercase{hina}.}
\and Muhammad Mohebujjaman\footrecall{uabm}\;\;\footrecall{support}
\and Taufiquar Khan\footremember{UNCC}{D\MakeLowercase{epartment of }M\MakeLowercase{athematics and} S\MakeLowercase{tatistics}, \MakeLowercase{and} C\MakeLowercase{enter for} T\MakeLowercase{rustworthy} A\MakeLowercase{rtificial} I\MakeLowercase{ntelligence through}
M\MakeLowercase{odel} R\MakeLowercase{isk} M\MakeLowercase{anagement} (TAIM\MakeLowercase{ing} AI), U\MakeLowercase{niversity of} N\MakeLowercase{orth} C\MakeLowercase{arolina} (UNC) \MakeLowercase{at} C\MakeLowercase{harlotte}, NC 28223, USA.}}

\maketitle

\begin{abstract} 
We develop a self-adaptive physics-informed neural network (PINN) framework that reliably solves forward Darcy flow and performs accurate permeability inversion in heterogeneous porous media. In the forward setting, the PINN predicts velocity and pressure for discontinuous, piecewise-constant permeability; in the inverse setting, it identifies spatially varying permeability directly from indirect flow observations. Both models use a region-aware permeability parameterization with binary spatial masks, which preserves sharp permeability jumps and avoids the smoothing artifacts common in standard PINNs. To stabilize training, we introduce self-learned loss weights that automatically balance PDE residuals, boundary constraints, and data mismatch, eliminating manual tuning and improving robustness, particularly for inverse problems. An interleaved AdamW–L-BFGS optimization strategy further accelerates and stabilizes convergence. Numerical results demonstrate accurate forward surrogates and reliable inverse permeability recovery, establishing the method as an effective mesh-free solver and data-driven inversion tool for porous-media systems governed by partial differential equations.
\end{abstract}




\pagestyle{myheadings}
\thispagestyle{plain}

\markboth{\MakeUppercase{Self-Adaptive PINN for Forward and Inverse Modeling of Porous Flow}}{\MakeUppercase{M. A. Aziz, T. Strauss, M. Mohebujjaman, and T. Khan}}

\section{Introduction} 

To present the new Physics-Informed Neural Network (PINN)-based forward and inverse problems, we consider the steady-state incompressible Darcy flow represented by the following Partial Differential Equations (PDEs) \cite{burman2022two,strauss2015statistical}:
\begin{align}
\bu + K(\bx)\nabla p &= \textbf{0} \hspace{10mm} \text{in } \Omega, \label{momentum}\\
\nabla \cdot \bu &= f(\bx) \hspace{4.5mm} \text{in } \Omega, \label{compress}\\
p &= p_B \hspace{7mm} \text{on } \Gamma_D,\label{dirichletbc} \\
\bu \cdot \bn &= u_B \hspace{7mm} \text{on } \Gamma_N,\label{neumanbc}
\end{align}
where $\bu$ is the Darcy velocity, \( p \) is the pressure, \( K \) is the spatially varying permeability, $\bx$ is the spatial variable, and $f(\bx)$ is the mass source,  \( \Omega  \) is the domain, and $\bn$ is the outward unit normal vector on the boundary. The $p_B$ and $u_B$ are the pressure on the Dirichlet boundary $\Gamma_D$ and the normal component of velocity on the Neumann boundary $\Gamma_N$, respectively.  The domain \( \Omega\) is partitioned into two regions with distinct permeability values. The \textit{Darcy law} that relates the fluid velocity and pressure gradient is represented by \eqref{momentum} and the \textit{conservation of mass} is expressed by \eqref{compress}. The pressure  Dirichlet boundary condition in \eqref{dirichletbc} ensures uniqueness of the pressure. In the forward problem, for a given $K(\bx)$, the goal is to solve for $\bu$ and \( p \), while in the inverse problem, the objective is to infer \( K(\bx) \) from observed flow data.

PINNs have emerged as a powerful framework for solving PDEs by embedding physical laws into neural network training through automatic differentiation and soft-penalty constraints \cite{raissi2019physics}. Since their introduction,
PINNs have been successfully applied to a wide range of forward and inverse problems in computational physics, including incompressible Navier–Stokes flow \cite{gu2024physics,jin2021nsfnets,oldenburg2022geometry}, two-phase porous media flow \cite{gasmi2021physics,hanna2022residual,yan2024physics,yang2023using,zhang2023physics}, and coupled multiphysics systems \cite{khadijeh2025multistage,sun2022physics}.

In porous media applications, PINNs have been used both to predict flow fields and to reconstruct spatially varying permeability structures, e.g., Gasmi and Tchelepi \cite{gasmi2021physics} demonstrated the effectiveness of PINNs in capturing sharp saturation fronts in two-phase
flow. Some other works have focused on inverse modeling, e.g., Berardi et al. \cite{berardi2025inverse} introduced an adaptive inverse PINN that infers transport parameters (including permeability-related quantities) from sparse observations by embedding the unknowns as trainable variables and adaptively weighting multiple loss components during training.

Recently, several innovations have advanced the capabilities of PINNs in solving computationally challenging PDEs. Murari et al.~\cite{murari2025rw-pinn} proposed \textit{Residual-Weighted PINNs (RW-PINNs)}, which dynamically scale PDE residuals to improve training stability in both forward and inverse reaction-diffusion problems, supported by rigorous convergence analysis and empirical results. \textit{Adaptive Residual Splitting PINNs (ARS-PINNs)}, introduced by Cao et al.~\cite{cao2025arspinn}, decompose the global PDE residual into multiple sub-residuals and assign adaptive weights to each. This strategy improves convergence and local accuracy by mitigating the dominance of global residuals in the loss landscape, yielding superior performance in complex and multiscale PDE settings compared to both global and pointwise balancing methods.

Shukla et al.~\cite{shukla2021parallel} developed \textit{Parallel PINNs}, including \textit{conservative PINNs (cPINNs)} and \textit{extended PINNs (XPINNs)}, which employ domain decomposition in space and in space--time, respectively. These methods enable parallel training, localized hyperparameter tuning, and improved scalability for multiscale and inverse problems on complex geometries. Similarly, \textit{Subregion-Identified PINNs (SI-PINNs)}, introduced by Wang et al.~\cite{wang2025si-pinn}, utilize distinct neural networks for different subregions---such as Navier--Stokes, Darcy, and their interface---trained on separate datasets. This region-specific decomposition allows SI-PINNs to reduce training cost while preserving or enhancing accuracy in problems dominated by interface dynamics.

Additionally, \textit{self-adaptive loss balancing} methods~\cite{mcclenny2023self} have gained traction, employing gradient norm--based or uncertainty-based mechanisms to automatically tune loss weights during training, thus reducing the need for manual hyperparameter tuning. Finally, hybrid optimization strategies that combine adaptive optimizers like Weight-decoupled Adaptive Moment Estimation (AdamW)~\cite{loshchilov2017decoupled} or Stochastic Gradient Descent (SGD)~\cite{robbins1951stochastic} with Limited-memory Broyden–Fletcher–Goldfarb–Shanno (L-BFGS) methods~\cite{lu2021deepxde} have been shown to accelerate convergence and enhance solution smoothness.

Despite these advances, three major challenges remain when applying PINNs to Darcy flow in heterogeneous porous media: (i) Reconstructing discontinuous or sharply varying permeability fields. (ii) Avoiding manual tuning of loss weights, which can hinder generalization and stability. (iii) Ensuring robust convergence in inverse problems with limited or noisy data.

\subsection{Significance of the work} To address these issues, we propose a self-adaptive PINN framework for modeling Darcy flow with piecewise-constant permeability. The main novelties of our work are the following:

\begin{itemize}
  \item Region-specific permeability parameterization: Instead of predicting a random spatially varying permeability field, the inverse model learns two scalar permeability values, \( K_1 \) and \( K_2 \), assigned through a predefined geometric mask---a binary function that partitions the spatial domain into subregions. This is similar in spirit to domain decomposition strategies used in PINNs, such as cPINNs and XPINNs, which divide the computational domain into space or space--time subregions to improve scalability, interpretability, and performance~\cite{shukla2021parallel}. This parameterization enhances interpretability, reduces overfitting near discontinuities, and leverages known geometric priors to guide the inverse problem toward meaningful physical estimates.

  \item Self-learned loss reweighting: We propose a trainable loss reweighting scheme that employs \textit{sigmoid-activated parameters} to dynamically modulate the relative influence of PDE residuals, boundary conditions, and data mismatch components during training. While prior work on adaptive PINNs—such as McClenny and Braga-Neto~\cite{mcclenny2023self}—focused primarily on forward problems, our method \textit{extends this adaptive balancing to a coupled forward--inverse setting}, where the forward flow field and inverse parameter (permeability) estimation are learned \textit{jointly}. This unified treatment is particularly effective in \textit{heterogeneous media}, where region-specific parameterization (e.g., binary permeability masks) benefits from flexible loss scaling, leading to \textit{enhanced convergence, improved numerical stability, and greater interpretability}.

  \item Hybrid optimization strategy: We implement a hybrid training algorithm that interleaves optimization steps from AdamW and L-BFGS to enhance convergence stability and efficiency, particularly when recovering spatially varying material parameters from indirect flow observations. AdamW—a decoupled variant of Adam~\cite{kingma2014adam} —provides rapid initial convergence and better generalization through separate weight decay, while L-BFGS, a quasi-Newton method, excels in fine-tuning by leveraging curvature information. Alternating between these two allows us to combine the strengths of both first- and second-order methods. This approach is especially beneficial for stiff inverse problems with sharp parameter discontinuities (e.g., layered permeability). Our strategy extends prior hybrid optimization frameworks by incorporating region-specific parameterization and self-learned loss reweighting. These enhancements enable robust convergence even in discontinuous settings and are supported by recent work showing that Adam+L-BFGS interleaving significantly improves PINN training by reducing loss landscape stiffness~\cite{rathore2024losslandscape}.

  \item Validation: We benchmark both forward and inverse PINNs predictions against the Finite Element Method (FEM) outcomes.
\end{itemize}

The remainder of this paper is organized as follows: Section \ref{prelim} introduces the mathematical preliminaries, variational formulation of \eqref{momentum}-\eqref{neumanbc} and physical setup. Section \ref{justification} presents the PINN architectures, loss formulation, and optimization strategies. Section \ref{numerical-exp} describes the data generation using FEM and evaluation methodology. The forward and inverse PINNs results, their comparison with analogous FEM outcomes and discussions are given in Section \ref{results-and-discussions}. Finally, Section \ref{conclusion} includes the conclusion and future research directions.

\section{Notations and mathematical preliminaries} \label{prelim}
We assume $\Omega\subset \mathbb{R}^d$  for $d\in\{2,3\}$ is a bounded Lipschitz domain, with boundary \( \partial \Omega = \Gamma_D \cup \Gamma_N \), where \( \Gamma_D \cap \Gamma_N = \emptyset \), and the measure \( \mu(\Gamma_D) > 0 \) to ensure well-posedness. The spatial variable $\bx=(x,y)$, and $\bx=(x,y,z)$ for 2D, and 3D problems, respectively. The usual $L^2(\Omega)$ norm and inner product are denoted by $\|.\|$ and $(.,.)$, respectively. The given source term \( f \in L^2(\Omega) \), and the permeability \( K(\bx) \geq K_{\min} > 0 \), where $K_{\min}:=\min\limits_{\bx\in\Omega}K(\bx)$. For the forward problem, we define the following natural spaces for Darcy flow problem:
\begin{align*}
      \bV: &= H(\mathrm{div}; \Omega) = \{ \bv \in (L^2(\Omega))^d : \nabla \cdot \bv \in L^2(\Omega) \},
     \bV^0: = \{ \bv \in \bV : \bv \cdot \bn = 0 \text{ on } \Gamma_N \},\;\text{and } Q: = L^2(\Omega),
\end{align*} where \( \bV \) ensures conformity with the flux field.

\subsection{Variational formulations for the mixed FEM of the forward problem} Multiply both sides of \eqref{momentum}, and \eqref{compress} by the test function $\bv\in\bV^0$, and $q\in Q$, respectively, integrate by parts, and apply boundary condition \eqref{dirichletbc} to obtain the following continuous weak formulation: Find \( (\bu, p) \in \bV^0 \times Q \) such that
\begin{align}
    (K^{-1}\bu, \bv) - (p, \nabla \cdot \bv) &= - \int_{\Gamma_D} p_B \, \bn \cdot \bv \, ds, && \forall \bv \in \bV^0,\label{cpde1} \\
    (\nabla \cdot \bu, q) &= (f, q), && \forall q \in Q. \label{cmass-conser}
\end{align}

To discretize the Darcy system \eqref{cpde1}-\eqref{cmass-conser}, we denote the conforming finite element spaces $\bV_h\subset\bV$, $\bV_h^0\subset\bV^0$, and $Q_h\subset Q$ based on an edge-to-edge triangulations of $\Omega$ with maximum diameter $h$. The discrete algorithm of the Darcy system reads: Find \( (\bu_h, p_h) \in \bV_h^0 \times Q_h \) such that
\begin{align}
    (K^{-1}\bu_h, \bv_h) - (p_h, \nabla \cdot \bv_h) &= - \int_{\Gamma_D} p_B \, \bn \cdot \bv_h \, ds, && \forall \bv \in \bV_h^0,\label{pde1} \\
    (\nabla \cdot \bu_h, q_h) &= (f, q_h), && \forall q_h \in Q_h, \label{mass-conser}
\end{align}
where $\bu_h$ and $p_h$ are the discrete velocity and pressure, respectively.


\section{Theoretical justification of the proposed framework}\label{justification}
This section provides a theoretical foundation for the proposed self-adaptive PINNs framework, focusing on the identifiability of the parametric inverse problem, the motivation behind trainable loss weighting, and the design of the interleaved optimization scheme.

\subsection{Identifiability of the parametric permeability model}

We consider the 2D inverse problem of recovering a piecewise constant permeability field 
\begin{align}
K(x, y) = 
\begin{cases}
K_1, & \text{if } (x, y) \in \Omega_1, \\
K_2, & \text{if } (x, y) \in \Omega_2,
\end{cases}\label{K-function}
\end{align}
for $K_1,K_2\in\mathbb{R}$, from full observations of the velocity field \( \bu_h(x, y) \) and pressure field \( p_h(x, y) \) on the domain \( \Omega \subset \mathbb{R}^2 \). We assume that \( \Omega_1 \subset \Omega \) is a known subregion (e.g., a rectangular inclusion), while the complementary region \( \Omega_2 := \Omega \setminus \Omega_1 \). Thus, the domain is partitioned as \( \Omega = \Omega_1 \cup \Omega_2 \), with \( \Omega_1 \cap \Omega_2 = \emptyset \). The fields \( (\bu_h, p_h) \) are assumed to satisfy the discrete Darcy system \eqref{pde1}--\eqref{mass-conser}.

\begin{lemma}[Identifiability of region-specific constants $K_1,K_2$]\label{identifiability}
Let $\Omega=\Omega_1\cup\Omega_2$ with $\Omega_1\cap\Omega_2=\emptyset$ (up to a set of measure zero),
and let $K(\bx)=K_1\chi_{\Omega_1}(\bx)+K_2\chi_{\Omega_2}(\bx)$ with unknown constants $K_1,K_2>0$.
Assume $(\bu,p)\in (H^1(\Omega))^2\times H^1(\Omega)$ are known a.e. in $\Omega$ and satisfy
Darcy's law
\begin{equation}\label{darcy_a.e.}
\bu(\bx)+K(\bx)\nabla p(\bx)=\mathbf{0}\qquad \text{a.e. in }\Omega.
\end{equation}
If for each $j\in\{1,2\}$ there exists a measurable set $E_j\subset \Omega_j$ with $|E_j|>0$ such that
$|\nabla p(\bx)|>0$ for a.e. $\bx\in E_j$, then $(K_1,K_2)$ is uniquely determined by \eqref{darcy_a.e.}.
\end{lemma}

\begin{proof}
Let $\bx\in \Omega$ be such that \eqref{darcy_a.e.} holds and $|\nabla p(\bx)|>0$. Taking the dot product of
\eqref{darcy_a.e.} with $\nabla p(\bx)$ yields
\[
0=\bu(\bx)\cdot \nabla p(\bx)+K(\bx)|\nabla p(\bx)|^2,
\]
and therefore
\begin{equation}\label{K_formula}
K(\bx)=-\frac{\bu(\bx)\cdot \nabla p(\bx)}{|\nabla p(\bx)|^2}
\qquad \text{for a.e. }\bx\in\{|\nabla p|>0\}.
\end{equation}
By assumption, $|\nabla p|>0$ a.e. on $E_j\subset \Omega_j$, hence \eqref{K_formula} determines $K(\bx)$ uniquely
for a.e. $\bx\in E_j$. Since $K(\bx)\equiv K_j$ a.e. on $\Omega_j$, this uniquely identifies $K_j$.
Doing this for $j=1,2$ yields uniqueness of $(K_1,K_2)$.
\end{proof}

The Lemma \ref{identifiability} provides a theoretical justification for the inverse problem setup: Under ideal data availability and perfect modeling, the region-specific parameterization of \( K(\bx) \) is uniquely identifiable.

\subsection{Self-adaptive loss weighting and gradient balancing}
\label{subsec:adaptive-weights}

In practice, training PINNs involves optimizing a composite loss function composed of multiple competing objectives: The PDE residual, boundary conditions, and data mismatch terms \cite{raissi2019physics}. Assigning equal weights often yields instability due to disparate gradient magnitudes across loss terms \cite{wang2021understanding}; adaptive weighting methods address this by balancing gradients \cite{chen2018gradnorm,mcclenny2023self}. To mitigate this, we introduce trainable weights \( \omega_i \) that dynamically adjust during training:
\begin{equation}
\omega_i = \sigma(\alpha_i) s_i, \quad i = 1, \dots, m,
\label{eq:trainable-loss}
\end{equation}
where \( \alpha_i \in \mathbb{R} \) are unconstrained trainable parameters, \( s_i > 0 \) are fixed scaling constants, \( \sigma \) is the sigmoid activation function, and \( m \) is the number of distinct loss terms. This formulation ensures that each weight remains strictly positive and bounded within \( (0, s_i) \), promoting stable optimization and avoiding manual hyperparameter tuning \cite{chen2018gradnorm, kendall2018multi,mcclenny2023self}.

\textbf{Forward PINN loss:} For the forward problem, where the permeability \( K \) is known, the total loss $\mathcal{L}_{\text{forward}}$ of the forward problem is defined as:
\begin{equation}
\mathcal{L}_{\text{forward}} = \omega_1  \mathcal{L}_{\text{PDE}} + \omega_2 \mathcal{L}_{\text{BC}}.
\label{eq:total-forward-loss}
\end{equation}
Here, \( \mathcal{L}_{\text{PDE}} \) enforces the governing Darcy momentum \eqref{momentum} and continuity \eqref{compress} equations across the interior of the domain by penalizing residuals computed at collocation points, while \( \mathcal{L}_{\text{BC}} \) penalizes violations of the Dirichlet and Neumann boundary conditions.

\noindent The individual loss terms are defined as:
\begin{align}
\mathcal{L}_{\text{PDE}}: &= \left\| \nabla \cdot \bu_\theta - f \right\|^2 + \left\| \bu_\theta + K \nabla p_\theta \right\|^2, \\
\mathcal{L}_{\text{BC}}: &= \left\| p_\theta - p_B \right\|^2_{\Gamma_D} + \left\| \bu_\theta \cdot \bn - u_B \right\|^2_{\Gamma_N}.
\end{align}

Here, $\bu_\theta$ and $p_\theta$ denote the velocity and pressure fields predicted by the forward PINNs model, where $\theta$ represents the trainable parameters of the neural network.

All loss terms are measured in $L^2$ \cite{raissi2019physics, bochev2009least}. This choice corresponds to a maximum-likelihood objective under Gaussian perturbations \cite{hastie2009elements} and yields smooth, stable gradients for first/second-order optimizers such as L-BFGS.

\textbf{Inverse PINN loss:} For the inverse problem, where the spatial permeability field $K$ is unknown but the noisy or sparse observations of the flow field are given, then the total loss function $\mathcal{L}_{\text{inverse}}$ of the inverse problem includes an additional data-fidelity term compared to the forward loss function:
\begin{equation}
\mathcal{L}_{\text{inverse}} = \omega_1  \mathcal{L}_{\text{PDE}} + \omega_2  \mathcal{L}_{\text{BC}} + \omega_3  \mathcal{L}_{\text{data}}.
\end{equation}

Here, \( \mathcal{L}_{\text{data}} \) penalizes the discrepancy between the model predictions and observed velocity and/or pressure data. The individual loss terms are defined as:
\begin{align}
\mathcal{L}_{\text{PDE}}: &= \left\| \nabla \cdot \bu_\phi - f \right\|^2 + \left\| \bu_\phi + K_\phi \nabla p_\phi \right\|^2, \\
\mathcal{L}_{\text{BC}}: &= \left\| p_\phi - p_B \right\|^2_{\Gamma_D} + \left\| \bu_\phi \cdot \bn - u_B \right\|^2_{\Gamma_N}, \\
\mathcal{L}_{\text{data}}: &= \left\| \bu_\phi - \bu_\theta \right\|^2 + \left\| p_\phi - p_\theta \right\|^2.
\end{align}

Here, $\bu_\phi$, $p_\phi$ and $K_\phi$ denote the velocity, pressure, and permeability fields predicted by the inverse PINN model, where $\phi$ represents the trainable parameters of the neural network.

This approach is conceptually related to adaptive loss weighting strategies such as GradNorm~\cite{chen2018gradnorm} and uncertainty-based multi-task learning~\cite{kendall2018multi}, but here it is embedded directly within a physics-constrained PINN framework. Similar ideas have also been explored in the context of PINNs~\cite{mcclenny2023self}, though the proposed method is tailored for coupled forward-inverse training under sharp permeability discontinuities.

\subsection{Effect of Velocity Regularity on Inverse Permeability Identification}

In the numerical experiments, finite element solutions of the Darcy system occasionally exhibit localized non-smooth features in the velocity magnitude $\|\bu_h\|$. These irregularities are primarily induced by discretization effects, such as element-wise polynomial approximations, mesh anisotropy, or post-processing of vector fields into scalar magnitudes, and do not represent physical discontinuities of the underlying flow.

The inverse problem considered in this work seeks to identify the permeability field $K(\bx)$ from observations of the flow governed by Darcy's law,
\begin{equation}
\bu = -K \nabla p, \qquad -\nabla \cdot (K \nabla p) = f \quad \text{in } \Omega.
\end{equation}
Importantly, the permeability enters the model through the pressure gradient and flux balance rather than through pointwise variations of the velocity magnitude. The inverse PINN minimizes a composite loss of the form
\begin{equation}
\mathcal{L} = \|\nabla \cdot (K \nabla p)\|_{L^2(\Omega)}^2
+ \|\bu + K \nabla p\|_{L^2(\Omega)}^2
+ \|\bu - \bu^{\text{obs}}\|_{L^2(\Omega)}^2,
\end{equation}
which enforces global consistency of Darcy’s law and mass conservation in an integral sense.

Localized numerical irregularities in $\|\bu_h\|$ typically occur on sets of small measure and do not significantly alter the $L^2$-based residuals driving the inverse optimization. From a stability perspective, inverse coefficient identification problems are known to depend primarily on global energy norms of the solution rather than pointwise regularity; see, for example, \cite{isakov2017inverse, engl1996regularization}. Consequently, small-scale oscillations or discretization-induced artifacts in the velocity field have limited influence on the recovered permeability.

Moreover, neural network approximations implicitly favor smooth solutions due to their continuous functional representation and optimization bias. This implicit regularization effectively filters non-physical, discretization-level features while preserving the dominant pressure gradients and flux structures responsible for parameter identifiability. As a result, the inverse PINN can recover accurate permeability fields even when the reference velocity data contain localized non-smoothness.

\subsection{Motivation for interleaved first- and second-order optimization}

Gradient-based optimization of PINNs is known to suffer from flat or stiff loss landscapes, particularly in inverse problems~\cite{rathore2024losslandscape,wang2021understanding}. First-order optimizers such as \textit{AdamW} are efficient in the early stages of training due to adaptive learning rates and momentum, but often fail to converge to highly accurate solutions \cite{raissi2019physics}. \textit{AdamW} differs from standard Adam by decoupling weight decay from gradient updates, improving convergence and generalization in overparameterized models~\cite{loshchilov2017decoupled}.

Quasi-Newton methods such as \textit{L-BFGS}, on the other hand, leverage curvature information to provide rapid convergence near local minima but are computationally expensive and less stable when used alone, particularly in large-batch or noisy-gradient settings~\cite{nocedal1980updating}. To exploit the complementary strengths of both optimizers, we alternate between short blocks of AdamW steps and L-BFGS refinements. This hybrid scheme accelerates early training via AdamW and improves final accuracy and smoothness through L-BFGS, mitigating stagnation in inverse PINNs with stiff dynamics or discontinuous coefficients.

Proposed approach is inspired by classical hybrid optimization regimes and aligns with recent PINNs literature~\cite{rathore2024losslandscape}, which shows that hybrid training improves loss landscape conditioning and accelerates convergence.

\section{Numerical validation and error estimation}\label{numerical-exp}

To evaluate the accuracy and effectiveness of the proposed PINNs framework, we conduct a series of numerical experiments comparing PINNs predictions to FEM solutions. Both the forward and inverse PINNs models are assessed in terms of their ability to replicate velocity and pressure fields, as well as recover spatially varying permeability. To this end, we solve the 2D Darcy flow system \eqref{momentum}-\eqref{neumanbc} using FEM and forward PINN, and then present the inverse PINN for predicting the heterogeneous permeability field. This section outlines the process for generating reference data using FEM, describes the error metrics used for evaluation, and presents detailed comparisons between the predicted and reference solutions through spatial error maps, contour plots, and pointwise metrics.

\subsection{Finite element discretization and reference solutions}

To validate the proposed PINN framework, we compute reference solutions for the Darcy flow problem \eqref{momentum}–\eqref{neumanbc} using the finite element method (FEM). These FEM solutions are used exclusively for benchmarking, visualization, and quantitative error analysis of the PINN results.

We employ the Raviart--Thomas element of order zero (RT0) for the velocity approximation and piecewise constant polynomial $\mathbb{P}_0$ elements for the pressure approximation to solve the weak formulation in \eqref{pde1}–\eqref{mass-conser}. The computational domain is the unit square $\Omega = [0,1]^2$, discretized using a structured triangular mesh, resulting in $240{,}800$ velocity degrees of freedom and $80{,}000$ pressure degrees of freedom. In all numerical experiments, we consider a constant unit source term $f(\bx)=1$ in $\Omega$.

The permeability field $K$ is prescribed as a piecewise constant function,
\begin{align}\label{k-function2}
K(x,y) =
\begin{cases}
0.1,  & \text{if } (x,y)\in[0.4,0.6]\times[0.4,0.6], \\
0.05, & \text{if } (x,y)\in \Omega \setminus [0.4,0.6]\times[0.4,0.6],
\end{cases}
\end{align}
representing a high-permeability inclusion embedded in a lower-permeability matrix. This discontinuous binary structure models heterogeneous porous media and provides a non-smooth reference field for assessing inverse reconstruction accuracy.

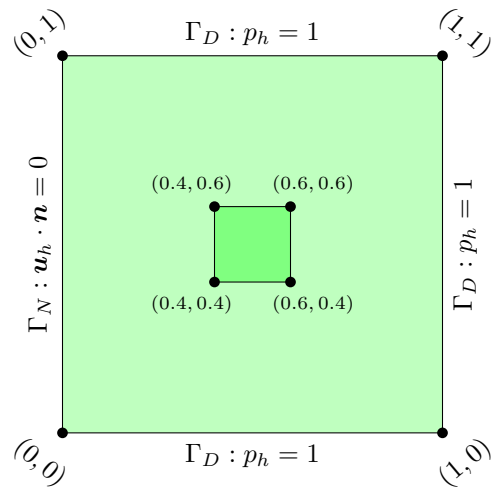
\begin{wrapfigure}[20]{r}[-1.0cm]
{0.35\textwidth}\centering
 \vspace{-\normalbaselineskip}
\begin{tikzpicture}
	\draw[fill=green!25] (0,0) rectangle (5,5);
	\draw[fill=green!50] (2,2) rectangle (3,3);

	\node at (2.5,-0.3) {$\Gamma_D:p_h=1$};
	\node at (5.3,2.5) {\rotatebox{90}{$\Gamma_D:p_h=1$}};
	\node at (2.5,5.3) {$\Gamma_D:p_h=1$};
	\node at (-0.3,2.5) {\rotatebox{90}{$\Gamma_N:\bu_h\cdot\bn=0$}};

	\node at (-0.3,-0.4) {\rotatebox{-40}{$(0,0)$}};
	\node at (5.3,-0.4) {\rotatebox{40}{$(1,0)$}};
	\node at (5.3,5.3) {\rotatebox{-40}{$(1,1)$}};
	\node at (-0.3,5.3) {\rotatebox{40}{$(0,1)$}};

	\node at (1.7,3.3) {\scriptsize{$(0.4,0.6)$}};
	\node at (3.3,3.3) {\scriptsize{$(0.6,0.6)$}};
	\node at (1.7,1.7) {\scriptsize{$(0.4,0.4)$}};
	\node at (3.3,1.7) {\scriptsize{$(0.6,0.4)$}};

	\fill (0,0) circle[radius=2pt];
	\fill (0,5) circle[radius=2pt];
	\fill (5,5) circle[radius=2pt];
	\fill (5,0) circle[radius=2pt];

	\fill (2,2) circle[radius=2pt];
	\fill (2,3) circle[radius=2pt];
	\fill (3,2) circle[radius=2pt];
	\fill (3,3) circle[radius=2pt];
\end{tikzpicture}
\caption{Schematic of the computational domain and imposed boundary conditions.}
\label{boundary-schematic}
\end{wrapfigure}

Mixed boundary conditions are imposed as follows: a Dirichlet boundary condition $p_h = p_B = 1$ is applied on the bottom, right, and top boundaries, while a homogeneous Neumann boundary condition $\bu_h \cdot \bn = 0$ (no-flux) is enforced on the left boundary. A schematic of the computational domain and boundary conditions is shown in Fig.~\ref{boundary-schematic}. The FEM implementation is carried out in the C++ finite element platform FreeFem++ \cite{MR3043640}, and the resulting sparse saddle-point system is solved using the direct solver UMFPACK \cite{davis2004algorithm}. The computed velocity and pressure fields $(\bu_h, p_h)$ are interpolated onto a uniform $101 \times 101$ Cartesian grid and exported as tabular data. These interpolated FEM solutions are used solely as reference data for visual comparison and quantitative error evaluation of the PINN predictions, and are treated as the benchmark throughout this work.



\textbf{Data format:} The exported data includes the coordinates \( (x, y) \), velocity components \( u_x, u_y \) for $\bu_h=( u_x, u_y)$, and pressure \( p_h \). This dataset is stored in a structured plain text file and used consistently for both forward model comparison and inverse permeability reconstruction.

\subsection{Experimental configuration}

The experimental setup used to train the forward and inverse PINNs models. All experiments were implemented in Python using PyTorch and executed on the Cheaha high-performance computing cluster at the University of Alabama at Birmingham. Computations were performed on NVIDIA Tesla P100 GPUs (Pascal architecture) with 16 GB memory, using CUDA 12.2 and cuDNN 8.9.2.26.

\textbf{Forward PINN architecture:}
We use a fully connected feedforward network with six hidden layers that maps spatial coordinates \((x,y)\) to the predicted fields \((\bu_{\theta},p_{\theta})\).
All hidden layers employ the GELU activation~\cite{hendrycks2016gaussian}, and residual (skip) connections are included to improve gradient flow.
Weights are initialized with Xavier initialization~\cite{glorot2010understanding} to maintain stable variance across layers.
Automatic differentiation (AutoDiff) is used to obtain all required spatial derivatives for the PDEs and boundary terms, and trainable adaptive weights \(\omega_i=\sigma(\alpha_i)s_i\) are employed to balance the multi-term loss during training. A schematic of the forward PINN is shown in the accompanying diagram \ref{forward-schematic}.
\begin{figure} [h!]
		\centering
		{\includegraphics[scale=.7]{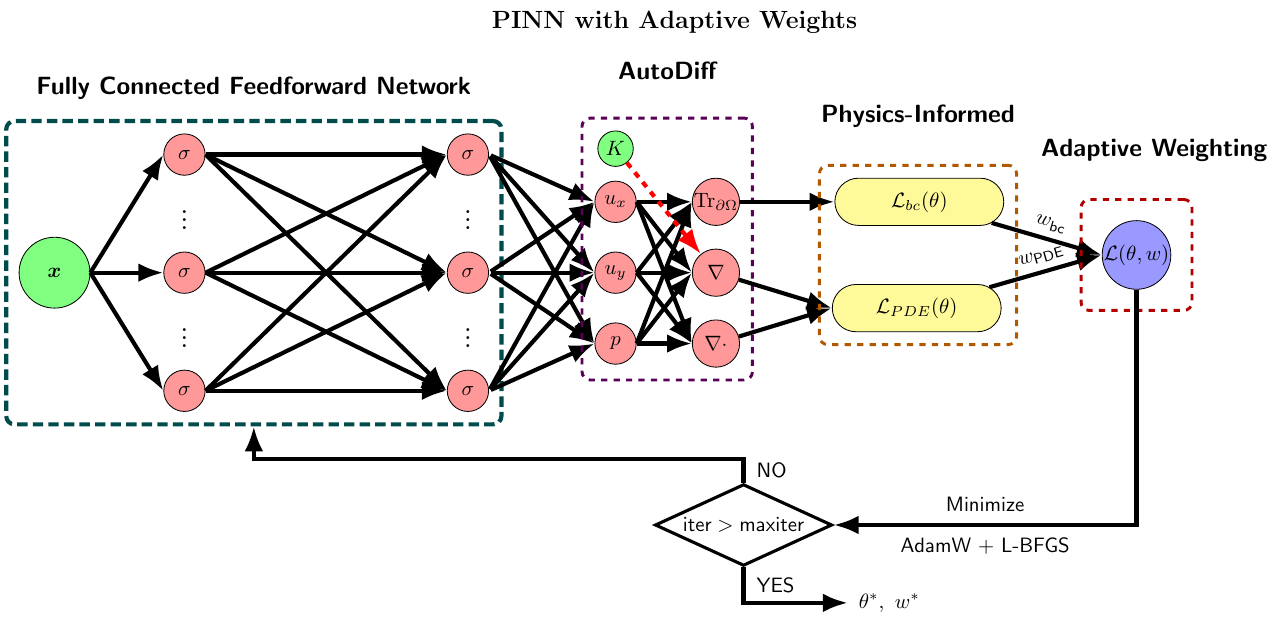}}
		
     \caption{Forward PINN schematic diagram.}
    \label{forward-schematic}
\end{figure}

\textbf{Inverse PINN Architecture:}
We use the same fully connected feedforward network with six hidden layers, mapping \((x,y)\) to \((\bu_\phi,p_\phi)\) and the inferred permeability \(K_1, K_2\). The same GELU activations, residual connections, and Xavier initialization are used. Automatic differentiation supplies spatial derivatives, and the same trainable adaptive weights \(\omega_i=\sigma(\alpha_i)s_i\) balance the multi-term loss. A schematic of the inverse PINN is shown in the accompanying diagram \ref{inverse-schematic}.

\begin{figure} [H]
		\centering
		{\includegraphics[scale=.7]{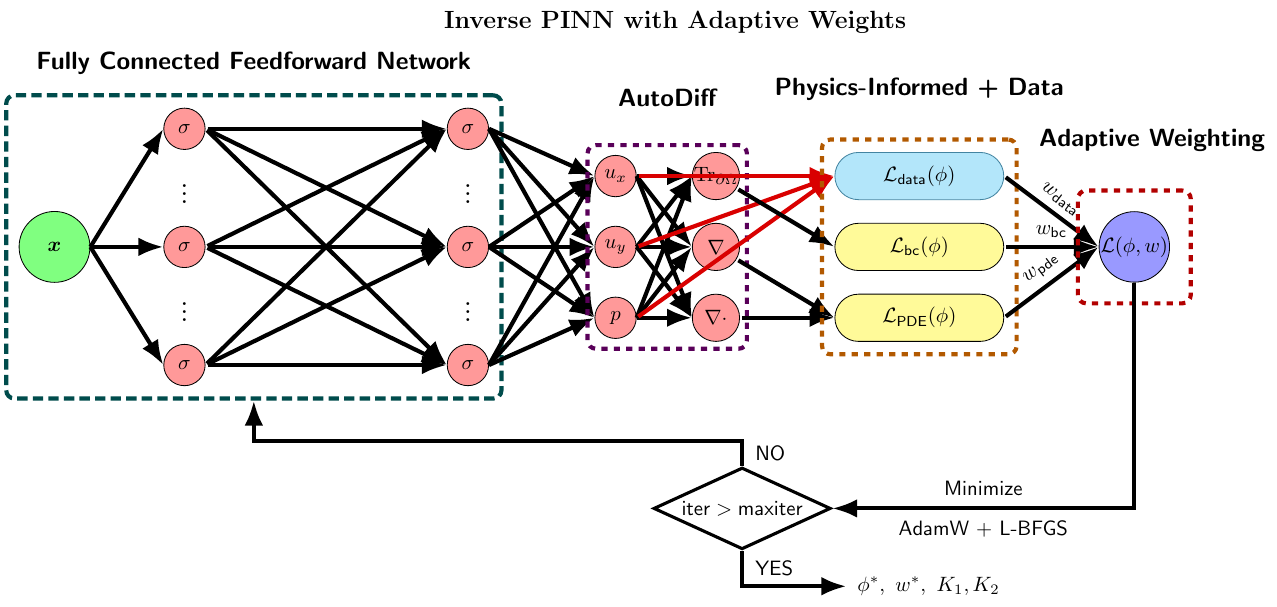}}
		
     \caption{Inverse PINN schematic diagram.}
    \label{inverse-schematic}
\end{figure}

\textbf{Adaptive loss weights:} We implement the self-adaptive loss weighting scheme described in Section ~\ref{subsec:adaptive-weights}, where each loss term is assigned a trainable weight of the form
\[
\omega_i = \sigma(\alpha_i) s_i,
\]
where \( \alpha_i \in \mathbb{R} \) are learnable parameters, \( \sigma \) is the sigmoid function, and \( s_i \) are fixed base scales. These weights are updated alongside the network parameters via backpropagation \cite{rumelhart1986learning}, enabling dynamic balancing of PDE residuals, boundary condition penalties \cite{raissi2019physics}, and data mismatch losses \cite{sun2022physics} throughout training.

\textbf{Collocation and boundary points:}
The training dataset consists of $10{,}201$ points on a uniform $101\times101$ Cartesian grid over $\Omega=[0,1]^2$, including boundary nodes. Points on $y=0$, $y=1$, and $x=1$ enforce Dirichlet conditions, points on $x=0$ enforce the Neumann condition, and the remaining points contribute to the interior PDE residual. This unified discretization imposes physics and boundary constraints simultaneously without excluding boundary nodes.

For the inverse problem, the same grid is reused, with each point assigned synthetic velocity and pressure values obtained from the trained forward PINN, providing full-field observations for permeability recovery.

\textbf{Optimization strategy:} Training is performed in two phases. Phase~I uses AdamW with learning rate $\eta=10^{-3}$, weight decay $10^{-6}$, gradient clipping at $2.0$, and a scheduler that halves $\eta$ after 700 stagnant epochs. Phase~II applies L-BFGS with strong Wolfe line search~\cite{nocedal2006numerical} in blocks of 100 iterations, interleaved with AdamW steps and using a history size of 50. Each model is trained for approximately $10{,}000$ epochs (6{,}000 AdamW epochs followed by 40 alternating AdamW--L-BFGS blocks).

\textbf{Implementation details:} The models are implemented in PyTorch with native automatic differentiation~\cite{paszke2019pytorch} to compute spatial derivatives in the loss terms. The entire pipeline, including data loading, model training, and visualization, is reproducible via open-source scripts and notebooks, which are available at \url{https://github.com/MAAziz23/Darcy-PINN}.

\section{Results and discussions} \label{results-and-discussions}
In this section, we represent and analyze the outcomes of the forward FEM, the forward PINN, and the inverse PINN for the Darcy flow problem we have considered.

\subsection{Forward PINN results}
\label{sec:forward-results}

We assess the forward PINN for the Darcy problem with known permeability $K(x,y)$, which predicts the velocity--pressure pair $(\bu_\theta, p_\theta)$ over the domain and is evaluated against a high-fidelity FEM reference. Specifically, we present (i) field-level comparisons of the velocity magnitude $\|\bu\|$ and pressure $p$ (forward PINN vs.\ FEM), (ii) pointwise absolute error maps for $u_x, u_y,$ and $p$, and (iii) aggregate errors measured by the root mean squared error (RMSE) for each output component, providing a coherent view of qualitative agreement and quantitative accuracy.\\

\noindent\textbf{Domain-wide speed field \(\|\bu\|\):}
With auto-scaled colorbars, the forward PINN and FEM speed fields show strong qualitative agreement, including elevated speeds along the top and bottom Dirichlet edges, attenuation toward the no-flux boundary at $x=0$, and a broad low-speed core near mid-height. Both solutions exhibit similar lateral contour bending due to the two-level permeability, while the PINN produces slightly smoother fine-scale features than FEM. Overall, the comparison confirms close agreement in global flow structure and boundary-layer behavior.

\begin{figure} [ht]
		\centering
		{\includegraphics[scale=.29]{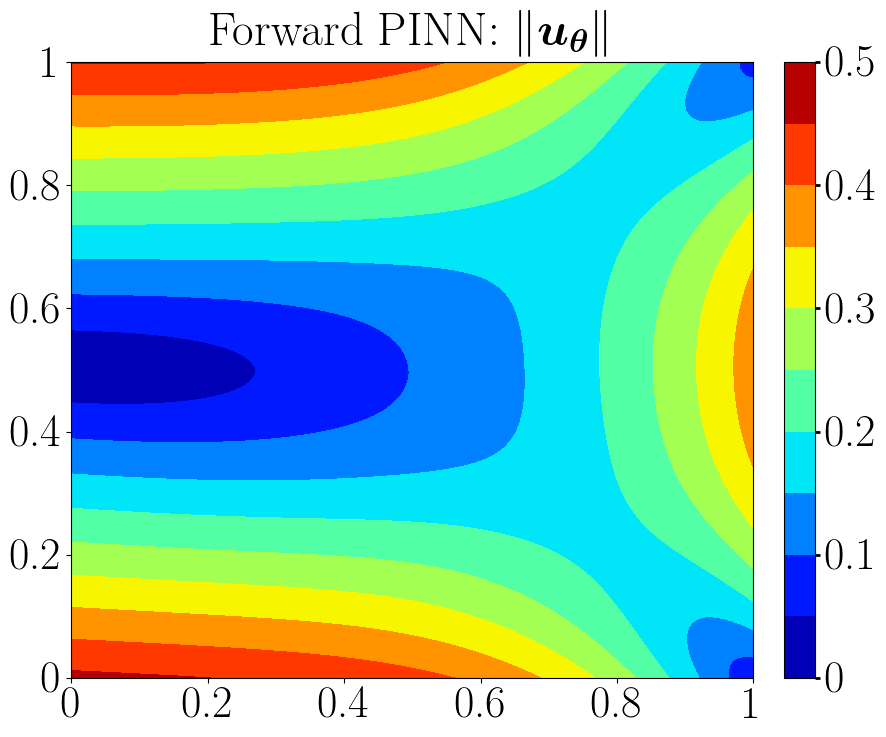}}\hspace{7mm}
		{\includegraphics[scale=.29]{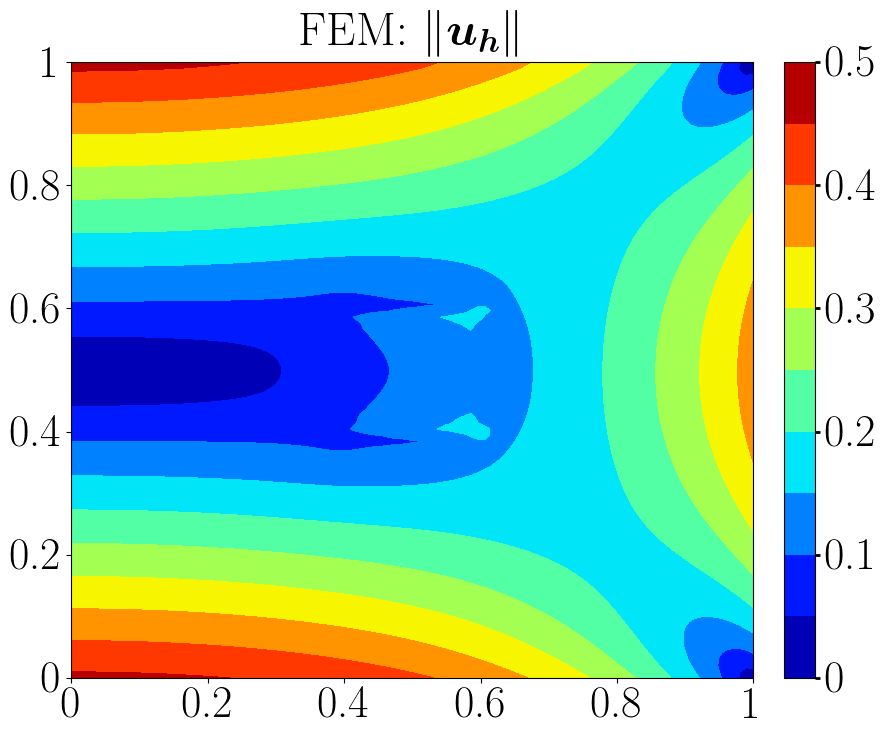}}
        \captionsetup{
    width=0.8\linewidth,
    justification=centering,
    singlelinecheck=false}
		\caption{\small Speed ($\|\bu\|$) comparison between the forward PINN (left) and the FEM reference solution (right).}
    \label{fig:for-velocity}
	\end{figure}


\textbf{Domain-wide pressure field \(p\):}
The pressure fields from the forward PINN and FEM show strong qualitative agreement, with a high-pressure lobe near the left boundary and smooth decay toward the right, consistent with the prescribed boundary and flux conditions. Similar iso-contour bending in the upper and lower halves indicates that the forward PINN captures the dominant pressure gradients. At finer scales, the PINN solution is slightly smoother, exhibiting marginally reduced peak amplitudes and more diffuse transitions near the right boundary. Overall, the comparison confirms close agreement in the global pressure distribution, with only minor smoothing of sharp features by the PINN.

\begin{figure} [ht]
		\centering
		{\includegraphics[scale=.29]{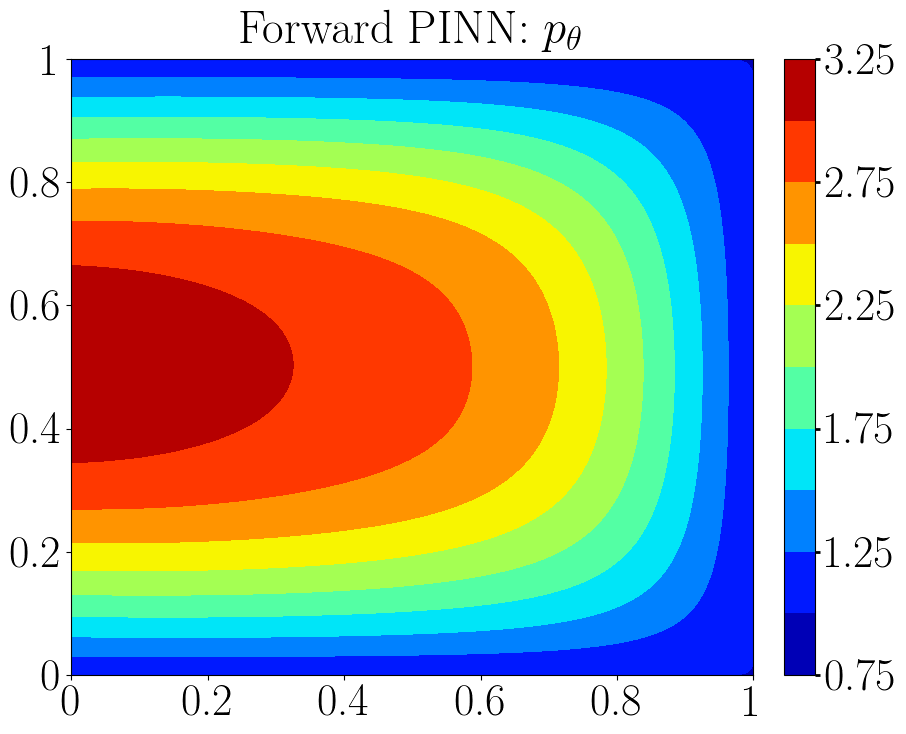}}\hspace{7mm}
		{\includegraphics[scale=.29]{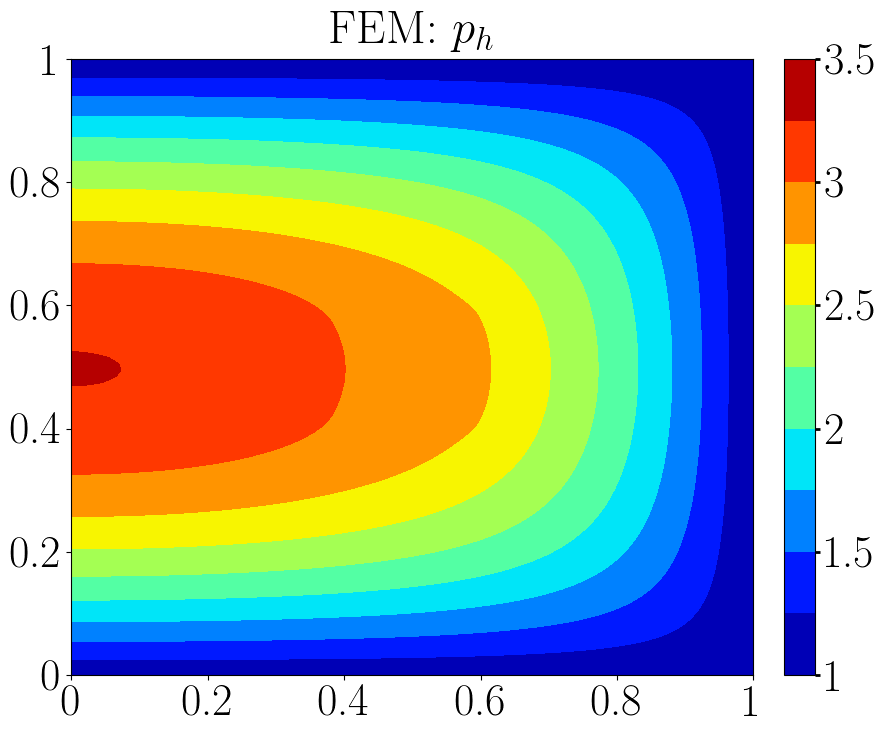}}
        \captionsetup{
    width=0.8\linewidth,
    justification=centering,
    singlelinecheck=false}
		\caption{\small Pressure field \( p(x, y) \) comparison between the forward PINN (Left) and FEM reference solution (Right).}
    \label{fig:for-pressure}
	\end{figure}


\noindent\textbf{Pointwise absolute-error maps for \(u_x,u_y,p\):}
We quantify local discrepancies using the absolute error
\begin{equation}    
AE_q(x,y)\;=\;\bigl|\,q^{\theta}(x,y)\;-\;q^{h}(x,y)\,\bigr|\,, 
\qquad q\in\{u_x,u_y,p\}.
\end{equation}

Figure~\ref{fig:Error-map-for} shows \(AE_q(x,y)\) evaluated at each grid point for \(q\in\{u_x,u_y,p\}\).
The velocity errors \((u_x,u_y)\) (left and middle) are small over most of the domain and are primarily
concentrated near the permeability interface, where gradients change abruptly, and in a thin strip adjacent
to the domain corners. The \(u_x\) panel exhibits a ring of elevated error encircling the interior high–low
transition, while the \(u_y\) panel shows narrow vertical bands consistent with the local shear. By contrast,
the pressure error (right) presents a broader, low-amplitude pattern that mirrors the large-scale curvature
of \(p\); the error remains modest away from boundaries and peaks mildly along the interior bend.
Overall, these maps indicate that discrepancies are localized around regions of strongest variation, with
velocity components showing sharper, interface-aligned features and pressure exhibiting a smoother,
domain-wide envelope of small error.
\begin{figure}[h!]
    \centering
    \includegraphics[width=0.9\textwidth]{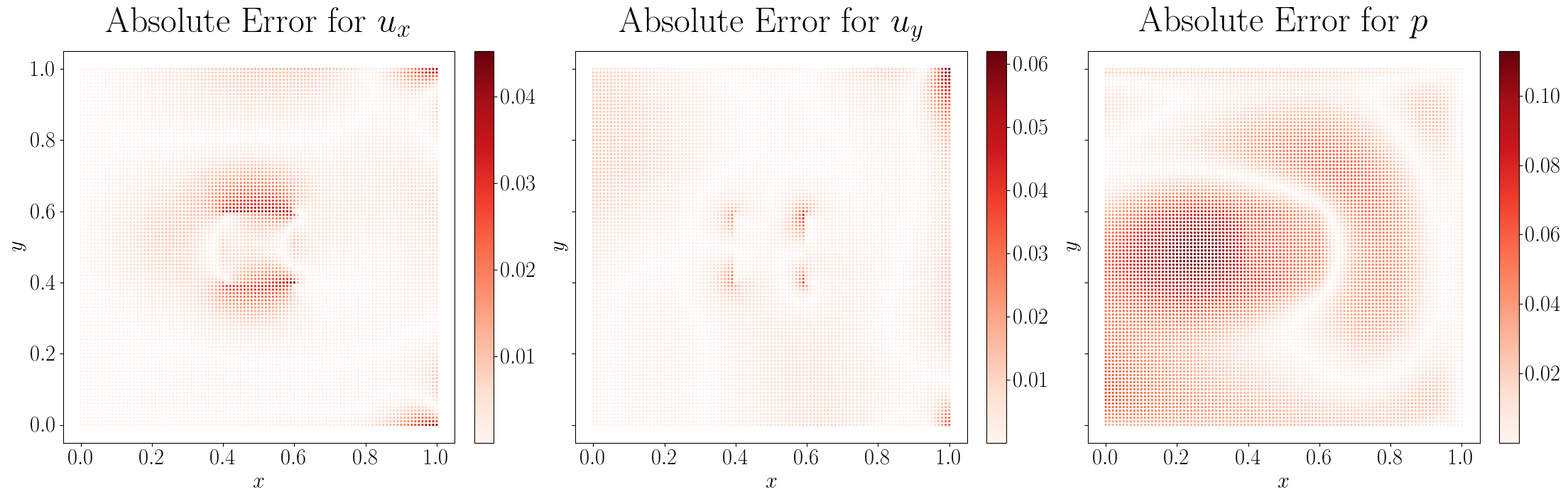}
    \captionsetup{
    width=0.8\linewidth,
    justification=centering,
    singlelinecheck=false}
    \caption{\small Pointwise absolute error maps for $u_x$ (left), $u_y$ (middle), and $p$ (right) computed
    between the forward PINN and FEM reference solution.}
    \label{fig:Error-map-for}
\end{figure}


\textbf{Aggregate error via RMSE.}
To quantify prediction accuracy across output fields, we report the root mean squared error (RMSE)
\begin{equation}
\mathrm{RMSE}_{\text{forward}}(q)
=\sqrt{\frac{1}{N}\sum_{i=1}^{N}\bigl(q_i^{\theta}-q_i^{h}\bigr)^2},
\qquad q\in\{u_x,u_y,p\},
\label{eq:mse_forward}
\end{equation}
where \(N=10,201\) is the total number of evaluation points, \(q_i^{\theta}\) is the Forward PINN prediction, and
\(q_i^{h}\) is the FEM reference at point \(i\). RMSE provides a single scalar, in the same units as \(q\),
for direct comparison across \(u_x,u_y\), and \(p\).

\begin{table}[h!]
\centering
\begin{tabular}{|c|c|c|c|}
\hline
 & \(u_x\) & \(u_y\) & \(p\) \\\hline
RMSE & 0.006880 & 0.006736 & 0.040034 \\\hline
\end{tabular}
\caption{RMSE of Forward PINN predictions relative to FEM for velocity components and pressure.}
\label{tab:MSE-for}
\end{table}

Taken together, the field-level comparisons (Figs.~\ref{fig:for-velocity}--\ref{fig:for-pressure}),
the pointwise absolute-error maps (Fig.~\ref{fig:Error-map-for}), and the RMSE metrics
(Table~\ref{tab:MSE-for}) indicate that the Forward PINN reproduces the FEM solution with high fidelity.
Discrepancies are localized primarily near the permeability interface and corners, while the rest of the
domain is visually indistinguishable. Quantitatively, the RMSEs are small relative to the field ranges
(velocity components $\sim 6\!\times\!10^{-3}$–$7\!\times\!10^{-3}$ versus speeds up to $\approx 0.5$; pressure
$\approx 4\!\times\!10^{-2}$ versus values up to $\approx 3.5$), underscoring close agreement in both
global organization and local detail.

The agreement between PINNs and FEM results confirms the effectiveness of our mesh-free approach in replicating complex porous flow behavior under known permeability conditions.


\subsection{Inverse PINN Results}
\label{sec:inverse-results}

We now assess the ability of the inverse PINN to recover the permeability field \(K(x,y)\) from observed pressure and velocity. The true permeability is piecewise constant with two regions as in \eqref{K-function}; the interface between \(\Omega_1\) and \(\Omega_2\) is known a priori, while the values \(K_1\) and \(K_2\) are inferred using physical constraints, boundary conditions, and the observed velocity–pressure fields from the forward PINN.

\medskip

\emph{Remark:} The inverse problem uses synthetic data generated by the forward PINN, which may lead to optimistic performance metrics. Applications to experimental data would require additional regularization and uncertainty quantification.

\medskip 

\textbf{Global comparison of the permeability field \(K(x,y)\):} The panels in Figure~\ref{fig:k-comparison} show the true field alongside the inverse PINN prediction over the full domain of the permeability field \(K(x,y)\). The
predicted map captures the two-region geometry and the contrast between the high- and low-permeability
zones, with the interface aligned to the prescribed partition. Visual agreement at the domain scale
indicates that the model has learned a field consistent with the assumed piecewise structure rather than
introducing spurious gradients or artifacts near the interface. This whole-domain view is important:
even if the region-wise constants are accurately estimated, an incorrect interface or smeared transition
would manifest as noticeable discrepancies in these plots.

\begin{figure} [ht]
		\centering
		{\includegraphics[scale=.45]{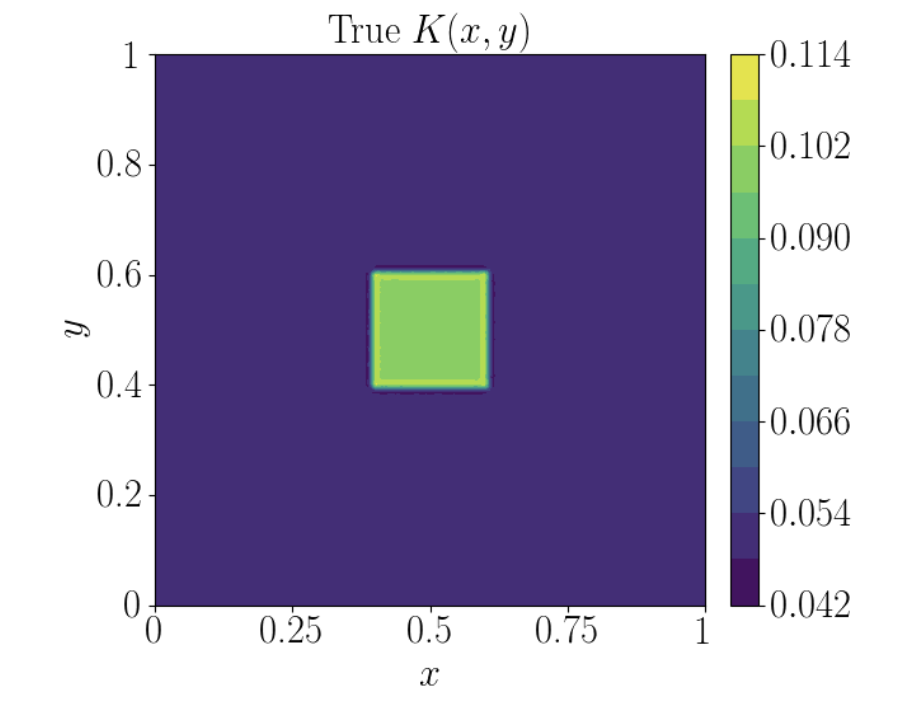}}\hspace{7mm}
		{\includegraphics[scale=.45]{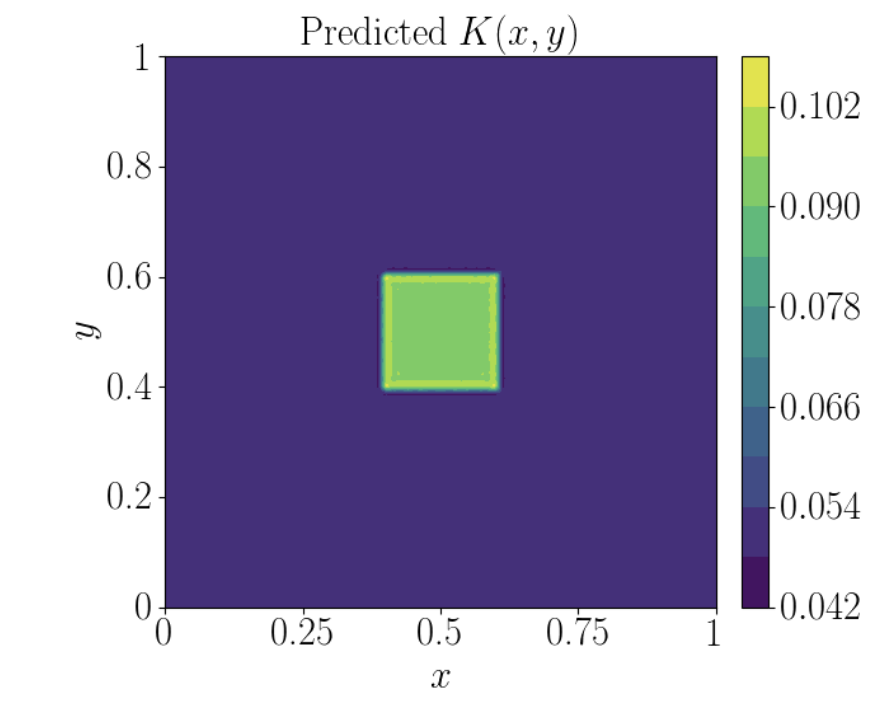}}          \captionsetup{
    width=0.8\linewidth,
    justification=centering,
    singlelinecheck=false}             \caption{\small Spatial comparison over the full domain: true versus inverse-PINN–predicted permeability fields \(K(x,y)\).
    The panels show that the inferred map reproduces the two-region geometry and the jump across the interface.}
    \label{fig:k-comparison}
\end{figure}

\noindent
\textbf{Local point verification of the permeability field \(K(x,y)\):}  To complement the global map, we report pointwise
checks at one interior reference location per region. Specifically, we evaluate \(K(x,y)\) at \((0.5,0.5)\),
representative of \(\Omega_1\), and at \((0.1,0.1)\), representative of \(\Omega_2\). Under the piecewise-constant
assumption, the true values at these points equal the corresponding region constants,
\(K(0.5,0.5)=K_1\) and \(K(0.1,0.1)=K_2\). Table~\ref{tab:K-values} lists the true constants next to the
inverse-PINN predictions at those reference points. Agreement here verifies that the network not only
reconstructs the correct interface at the field level but also attains the correct constant levels within
each subregion away from the interface.

\begin{table}[ht]
\centering
\begin{tabular}{|c|c|c|c|}
\hline
 & Reference point & True & Predicted \\
\hline
\( K_1 \) & \((0.5,\,0.5)\) & 0.10 & 0.095945 \\\hline
\( K_2 \) & \((0.1,\,0.1)\) & 0.05 & 0.050259 \\
\hline
\end{tabular}
\caption{Region-wise point verification for \(K(x,y)\): true constants versus inverse-PINN predictions at one interior reference point per region. The local checks corroborate the whole-domain agreement observed in Figure~\ref{fig:k-comparison}.}
\label{tab:K-values}
\end{table}


\noindent\textbf{Domain-wide speed field \(\|\bu\|\):}  Figure~\ref{fig:inv-velocity} shows that, using the permeability recovered by the inverse model, the induced speed field \(\|\bu\|\) shows
strong qualitative agreement with the FEM reference. Both solutions exhibit elevated speeds along the
Dirichlet boundaries, attenuation toward the no-flux edge at \(x=0\), and a broad low-speed core near
mid-height. The interface-driven bending of contours on the right is reproduced, indicating that the
estimated \(K_\phi\) yields the correct flow organization. Minor differences appear as smoother interior
features in the inverse result, but the overall structure and boundary-layer behavior are closely aligned.

\begin{figure} [ht]
		\centering
		{\includegraphics[scale=.29]{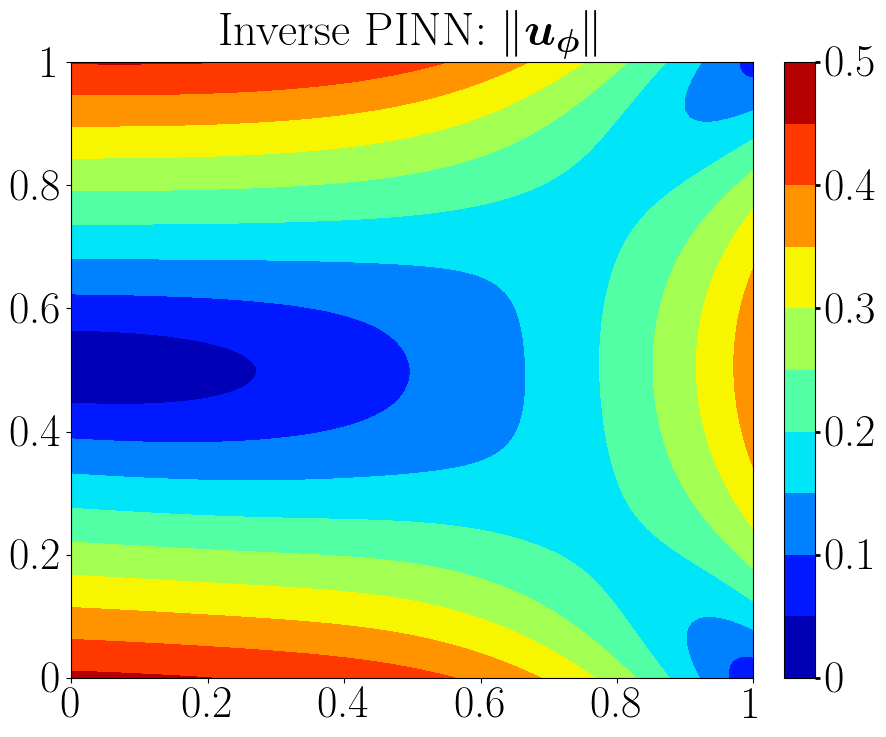}}\hspace{7mm}
		{\includegraphics[scale=.29]{Figures/Velocity_FreeFEM.png}}
        \captionsetup{
    width=0.8\linewidth,
    justification=centering,
    singlelinecheck=false}
		\caption{\small Speed ($\|\bu\|$) comparison between the inverse PINN (left) and the FEM reference solution (right).}
     \label{fig:inv-velocity}
	\end{figure}

\noindent\textbf{Domain-wide pressure field \(p\): Inverse PINN vs.\ FEM.}
Figure~\ref{fig:inv-pressure} shows that, using the permeability recovered by the inverse model, the induced
pressure field closely matches the FEM reference across the domain. Both solutions exhibit a high-pressure
lobe near the left boundary with a smooth monotonic decay toward the right, and the iso-contour curvature
is aligned in the upper and lower halves. The inverse result is slightly smoother near the interior peak and
along the right boundary layer, but no systematic bias is observed. Overall, the inverse PINN reproduces the
global pressure distribution and gradients consistent with the governing physics and the estimated permeability.

\begin{figure} [ht]
		\centering
		{\includegraphics[scale=.29]{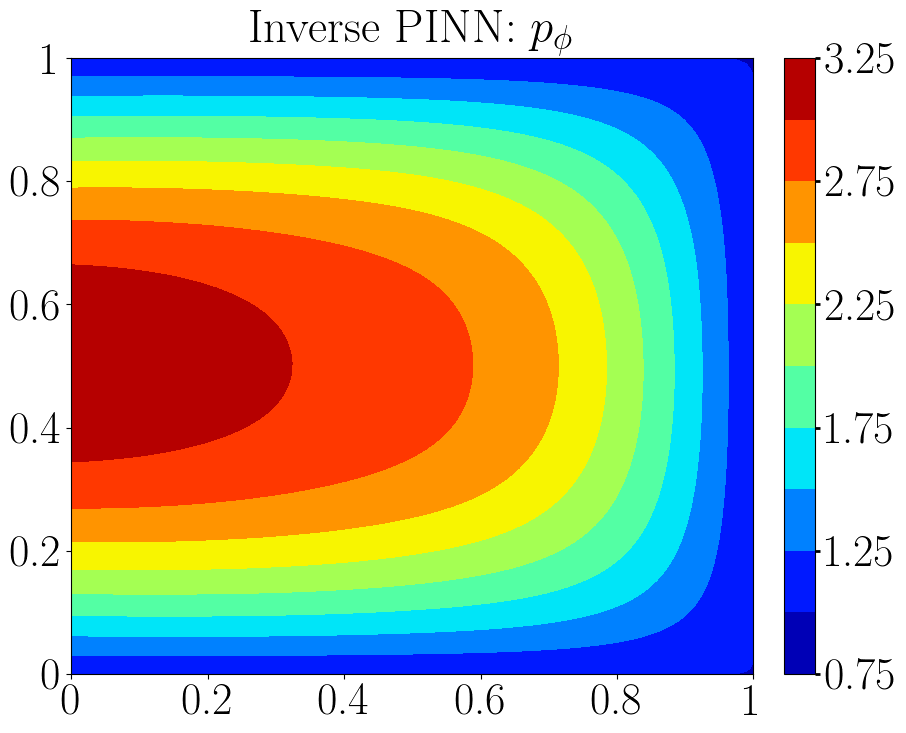}}\hspace{7mm}
		{\includegraphics[scale=.29]{Figures/Pressure_FreeFEM.png}}
        \captionsetup{
    width=0.8\linewidth,
    justification=centering,
    singlelinecheck=false}
		\caption{\small Pressure field \( p(x, y) \) comparison between the inverse PINN (Left) and FEM reference solution (Right).}
    \label{fig:inv-pressure}
	\end{figure}


\textbf{Pointwise absolute error maps:}
Figure~\ref{fig:Error-map-inv} visualizes the absolute error for \(u_x, u_y, p, \text{ and } K\). For the variables \(u_x,u_y,p\);
the errors are small over most of the domain and concentrate near the interior interface,
where gradients change most rapidly; faint bands at the outer boundaries are also visible.
The \(u_x\) map forms a ring around the interface, while \(u_y\) shows narrow vertical
streaks consistent with local shear. The pressure error presents a broader, low-amplitude
pattern that mirrors the large-scale curvature of \(p\).
For the parameter field \(K\), the error is sharply localized to the
interior square and decays to near-zero elsewhere, indicating accurate recovery of the piecewise-constant permeability and a well-resolved jump across the interface.
Together, these maps show that discrepancies are localized to regions of strongest variation, supporting the consistency of the recovered \(K\) and the induced state fields.

\begin{figure}[h!]
    \centering
    \includegraphics[width=0.9\textwidth]{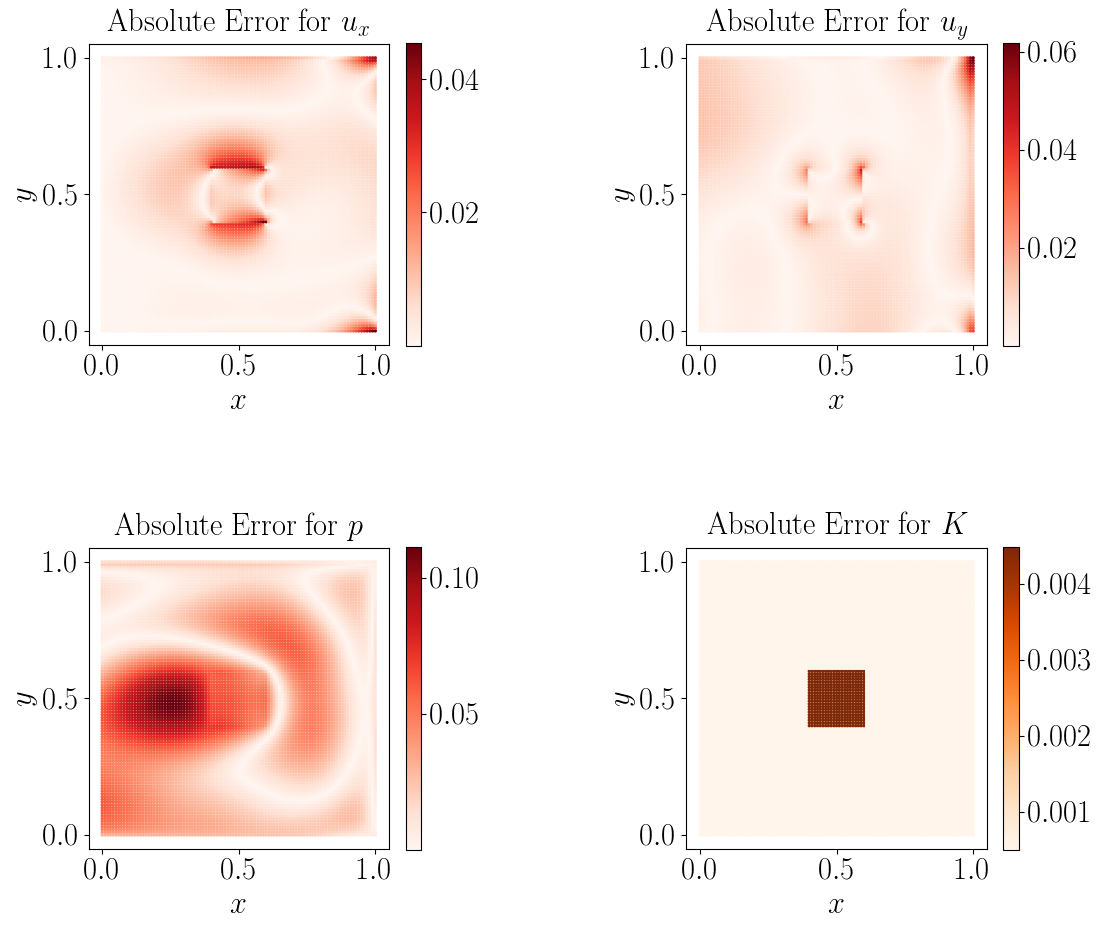}
    \captionsetup{
    width=0.8\linewidth,
    justification=centering,
    singlelinecheck=false}
    \caption{\small Pointwise absolute error maps for $u_x$ (top-left), $u_y$ (top-right), and $p$ (bottom-left) between the inverse PINN and the FEM reference solution, and for $K$ (bottom-right) with respect to the true permeability field.}
    \label{fig:Error-map-inv}
\end{figure}



\textbf{Aggregate error (inverse) via RMSE.}
The root mean squared error (RMSE) for the inverse PINN—computed against the FEM reference for each field—is summarized in Table~\ref{tab:MSE-inv}. 
RMSE is reported in the native units of each quantity and complements the pointwise error maps by providing a single, domain-wide measure of accuracy. 
Consistently small RMSEs for \(u_x,u_y\), and \(p\) indicate close agreement with the reference solution across the domain, while the low RMSE for \(K\) reflects accurate recovery of the piecewise-constant permeability (including the jump across the interface). 
Together with the visual comparisons, these values confirm that the inferred parameters yield state fields that are physically and numerically consistent with FEM.

\begin{table}[H]
\centering
\begin{tabular}{|c|c|c|c|c|}
\hline
 & \(u_x\) & \(u_y\) & \(p\) & \(K\) \\\hline
RMSE & 0.006881 & 0.006734 & 0.039922 & 0.001052 \\\hline
\end{tabular}
\caption{RMSE of inverse PINN predictions relative to FEM for velocity components, pressure, and permeability.}
\label{tab:MSE-inv}
\end{table}

\section{Conclusion and future research directions}\label{conclusion}

In this work, we proposed a robust PINN framework for solving both forward and inverse Darcy flow problems in porous media. Our method integrates several key innovations: (i) a self-adaptive loss weighting strategy to dynamically balance competing loss terms, (ii) skip connections and GELU activations to enhance gradient flow and expressivity, and (iii) a hybrid optimization scheme that interleaves AdamW and L-BFGS to stabilize training and improve convergence, particularly for stiff inverse problems with discontinuous coefficients.

The forward PINN successfully approximates velocity and pressure fields given a known permeability distribution, while the inverse PINN effectively recovers spatially varying permeability from indirect flow observations. Experimental results demonstrate high predictive accuracy across all physical variables, including the discontinuous permeability field. Notably, the use of self-learned loss weights eliminated the need for manual hyperparameter tuning and significantly improved training stability.


Future research directions include several promising avenues:

\begin{itemize}
  \item \textit{Extension to time-dependent flows:} Incorporating transient terms into the PINN formulation to handle unsteady Darcy or Navier-Stokes \cite{berry2025efficient} problems.
  
  \item \textit{Uncertainty quantification:} Integrating Bayesian PINNs or ensemble approaches \cite{raveendran2025efficient} to quantify prediction confidence, especially for inverse reconstructions.
  
  \item \textit{Multiscale modeling:} Applying the framework to multiscale or heterogeneous media, potentially using domain decomposition or localized basis functions.
  
  \item \textit{Generalization across domains:} Investigating how well trained PINNs can transfer or adapt to new geometries or boundary conditions via meta-learning or fine-tuning.

  \item \textit{Robustness to real data:} The current study uses synthetic data for inversion validation. Future work will test the framework on experimental or field datasets with noise and missing values, and incorporate uncertainty quantification methods to improve reliability in practical applications.
\end{itemize}

    
\section{Acknowledgment} The NSF is acknowledged for supporting this research through the grant DMS-2425308. We are grateful for the generous allocation of computing time provided by the Cheaha high-performance computing cluster at the University of Alabama at Birmingham.


\bibliographystyle{plain}


\bibliography{refs}
\end{document}